\documentclass[submission,copyright,noderivs]{eptcs}
\providecommand{\event}{NCMA 2022} % Name of the event you are submitting to
\usepackage{underscore}           % Only needed if you use pdflatex.

\usepackage{graphicx} 
\usepackage{mathrsfs}
\usepackage{amssymb}
\usepackage{amsmath}
\usepackage{amsthm}
\usepackage{pifont}
\usepackage[dvipsnames]{xcolor}
\usepackage{refcount}

\theoremstyle{plain}
\newtheorem{theorem}{Theorem}
\newtheorem{proposition}[theorem]{Proposition}
\newtheorem{lemma}[theorem]{Lemma}

\theoremstyle{definition}
\newtheorem{example}[theorem]{Example}
\newtheorem{definition}[theorem]{Definition}

\newcommand{\lfam}{\mathscr{L}}
\newcommand{\dfa}{\textrm{DFA}}
\newcommand{\cak}[1]{\textrm{DCA}(#1)}
\newcommand{\revcak}[1]{\textrm{REV-DCA}(#1)}
\newcommand{\invdelta}{\delta^{\scriptscriptstyle\leftarrow}}
\newcommand{\invdeltap}{\delta^{'\scriptscriptstyle\leftarrow}}
\newcommand{\invvdashm}{\mathrel{\vdash_M^{\protect\raisebox{2pt}{$\scriptscriptstyle\leftarrow$}}}}
\newcommand{\invvdashmp}{\mathrel{\vdash_{M'}^{\protect\raisebox{2pt}{$\scriptscriptstyle\leftarrow$}}}}
\newcommand{\leftend}{\mathord{\vartriangleright}}
\newcommand{\rightend}{\mathord{\vartriangleleft}}
\newcommand{\dollar}{\texttt{\$}}
\newcommand{\valc}{\textrm{VALC}}
\newcommand{\eoe}{\hspace*{\fill} $\blacksquare$\smallskip}
\DeclareMathOperator{\empt}{\bot}
\title{Reversible Computations of One-Way Counter Automata}

\author{Martin Kutrib and Andreas Malcher
\institute{%
  Institut f\"ur Informatik, Universit\"at Giessen\\
  Arndtstr.~2, 35392 Giessen, Germany}
\email{$\{$kutrib,andreas.malcher$\}$@informatik.uni-giessen.de}
}
 
\def\titlerunning{Reversible Computations of One-Way Counter Automata}
\def\authorrunning{M.~Kutrib, A.~Malcher}

\begin{document}

\maketitle

\begin{abstract}
Deterministic one-way time-bounded multi-counter automata are studied with 
respect to their ability 
to perform reversible computations, which means that the automata are also backward
deterministic and, thus, are able to uniquely step the computation back and forth. 
We study the computational capacity of such devices and obtain separation results
between irreversible and reversible $k$-counter automata for superpolynomial time.
For exponential time we obtain moreover an infinite and tight 
hierarchy with respect to the number of counters.
This hierarchy is shown with Kolmogorov complexity
and incompressibility arguments. In this way, on passing
we can prove this hierarchy also for ordinary counter automata.
This improves the known hierarchy for ordinary counter automata
in the sense that here we consider a weaker acceptance condition.
Then, it turns out that $k+1$ reversible counters are not better
than $k$ ordinary counters and vice versa.
Finally, almost all usually studied decidability questions turn out to be undecidable
and not even semidecidable for reversible multi-counter automata, if at least two
counters are provided. 
\end{abstract}

\section{Introduction}

In the last years, reversible computational models have earned a lot of attention.
The reversibility of a computation basically means that every configuration
has at most one unique successor configuration and at most one unique
predecessor configuration.
One incentive to study such computational devices performing 
logically reversible computations is probably the question posed by
Landauer of whether logical irreversibility is an unavoidable feature
of useful computers. Landauer has demonstrated the
physical and philosophical importance of this question
by showing that whenever a physical computer throws
away information about its previous state it must generate a corresponding 
amount of entropy that results in heat dissipation
(see~\cite{Bennet:1973:lrc} for further details and references).
First investigations on reversible computations have been started 
in the sixties of the last century both for Turing machines as well as
for the massively parallel model of cellular automata.
For both models it is known that irreversible computations can be made
reversible. 
For Turing machines it is shown in the work of Lecerf~\cite{lecerf:lmmtr:1963}
and Bennett~\cite{Bennet:1973:lrc} that for every Turing machine an equivalent 
reversible Turing machine can be constructed. 
For cellular automata it is known from~\cite{Morita:1995:rsodica} that every, 
possibly irreversible,
one-dimensional cellular automaton can always be simulated by a reversible 
one-dimensional cellular automaton in a constructive way.

At the other end of the Chomsky hierarchy there are the regular languages.
Here, Angluin has introduced
reversible computations in deterministic finite automata $(\dfa$) and 
showed that reversible $\dfa$s are weaker than $\dfa$s in general~\cite{Angluin:1982:irl}.
Moreover, it is known that two-way $\dfa$s and reversible two-way $\dfa$s
are equally powerful~\cite{kondacs:1997:pqfsa}. 
Recent results on reversible regular languages 
concern the descriptional complexity and the minimality of 
reversible (one-way) $\dfa$s and are obtained  
in~\cite{holzer:2018:mrdfa,lavado:2017:marra,lavado:2019:crora}.
Furthermore, due to their nature (real-time) quantum finite automata can be
said to be inherently reversible~\cite{bertoni:2001:rlaqfa,kondacs:1997:pqfsa}.
They do not capture the regular languages either. This is in contrast to
modified recent definitions that lead to quantum finite automata
which recognize all and only the regular languages with bounded 
error~\cite{hirvensalo:2010:qaote,say:2014:qfaami,yilmaz:2022:ecrlr}.
See~\cite{mereghetti:2021:qfafttp} and the references therein for a recent
survey on quantum automata.

For deterministic pushdown automata, the reversible variant has been introduced
in~\cite{kutrib:2012:rpa}, where its is in particular shown that the
reversible variant is weaker than the general one.
A special case of deterministic pushdown automata are deterministic one-counter
automata where the pushdown alphabet consists of one symbol only, apart from
the bottom symbol. Hence, the pushdown store can only be used to count a number of
symbols and no longer to store a sequence of different symbols.
In general, multi-counter automata are finite-state automata equipped with multiple counters which can
be incremented, decremented, and tested for zero.
It is well known that general one-way deterministic two-counter automata are computationally universal, that is, they can
simulate Turing machines~\cite{Minsky:1961:rupptotttm}. However, the latter simulation may need
an unbounded amount of space. Hence, deterministic space-bounded, as well as time-bounded, multi-counter automata
have been considered in~\cite{fischer:1968:cmacl} where, in particular, the case when the available time 
is restricted to real-time is studied. The authors establish in this case an
infinite and strict counter hierarchy as well as positive and negative closure results. The generalization to multi-counter
automata that may work nondeterministically as well as may use two-way motion on the input tape has
been done by Greibach~\cite{greibach:1976:rotconcl}. 
Recent results on one-way deterministic multi-counter automata are given by Petersen
in~\cite{petersen:2011:sbtbcm} where, in particular, some hierarchy results 
of Greibach concerning counters and polynomial time could be improved and tightened
at the price of a stronger acceptance condition than defined in~\cite{greibach:1976:rotconcl}.
Finally, we already mentioned that one-counter automata can be seen as a special
case of pushdown automata. Hence, multi-counter automata may be considered a 
special case of
multi-pushdown automata introduced and studied in~\cite{Breveglieri:1996:mpdlag}.

In this paper, we will consider reversible multi-counter automata. Such automata
have been investigated by Morita in~\cite{morita:1996:urtcm} with respect to
universal computations. In detail, the universality result of Minsky could
be improved, 
namely, it is shown by Morita that any Turing machine can already be simulated 
by a \emph{reversible} two-counter automaton. It should be noted that,
naturally, the simulation
of Minsky as well as the reversible simulation of Morita may need
an unbounded amount of space and time. In addition, the input has to be provided 
suitably encoded by using prime numbers. In this paper, we will therefore
consider time-bounded (and hence space-bounded) reversible multi-counter automata 
that process a given plain unencoded input. The paper is organized as follows.
The definition of the model and illustrating examples are given in Section~\ref{sec:prelim}.
In Section~\ref{sec:compcap} we study the computational capacity in detail and 
obtain as first result that there is a regular language that can clearly be accepted
by irreversible $k$-counter automata in real-time, for any $k\geq 0$, but
cannot be accepted by any \emph{reversible} $k$-counter automaton within 
time~$2^{o(n)}$, regardless of the number of counters.
We then prove a tight counter hierarchy for reversible counter automata working in
exponential time. 
This hierarchy is shown with Kolmogorov complexity
and incompressibility arguments. In this way, on passing
we can prove this hierarchy also for ordinary counter automata.
This improves the known hierarchy for ordinary counter automata~\cite{petersen:2011:sbtbcm}
in the sense that here we consider a weaker acceptance condition.
Finally, we have
incomparability results between reversible and irreversible counter automata
if the reversible automata have strictly more counters than 
the irreversible once.
Hence, we can draw a complete picture
of the relations between the language families discussed. 
In Section~\ref{sec:deci} we investigate decidability questions for reversible counter
automata. It turns out that all usually studied questions such as, for example,
emptiness, finiteness, infiniteness, inclusion, and equivalence are undecidable
and not even semidecidable for reversible counter automata with at least two counters.

\section{Preliminaries}\label{sec:prelim}

We denote the non-negative integers $\{0,1,2,\dots\}$ by $\mathbb{N}$.
Let $\Sigma^*$ denote the set of all words over the finite alphabet $\Sigma$.
We write $\lambda$ for the \emph{empty word},
and let $\Sigma^+ = \Sigma^* \setminus \{\lambda\}$.
The set of words of length at
most $n\geq 0$ is denoted by $\Sigma^{\leq n}$. 
The \emph{reversal} of a word $w$ is denoted by $w^R$.  For the \emph{length} of $w$, we
write~$|w|$. The number of occurrences of a symbol $a \in \Sigma$ in $w \in \Sigma^*$ is
written $|w|_a$. We use $\subseteq$ for \emph{inclusions} and~$\subset$ for \emph{strict inclusions}.

Let $k \geq 0$ be an integer. A one-way $k$-counter 
automaton is a finite automaton having a single read-only input tape whose
inscription is the input word in between two endmarkers (we provide two
endmarkers in order to have a definition consistent with two-way devices). 
In addition, it is equipped with~$k$ counters.
At the outset of a computation the counter automaton is in the designated 
initial state, the counters are set to zero, and the head of the input tape 
scans the left endmarker.
Dependent on the current state, the currently scanned input symbol, and
the information whether the counters are zero or not, the counter automaton 
changes its state, increases or decreases the counters, and moves the 
input head one cell to the right or not. 
The automata have no extra output tape but
the states are partitioned into accepting and rejecting states.

\begin{definition}\label{def:dca-k}
A \emph{deterministic one-way counter automaton with $k\geq 0$ counters
(abbreviated as $\cak{k}$)} is a system 
$M=\langle Q, \Sigma,k,\leftend,\rightend, \delta, q_0, F\rangle$, where
\begin{enumerate}
\item
$Q$ is the finite set of \emph{internal states}, 
\item
  $\Sigma$ is the finite set of \emph{input symbols},
\item 
$k\geq 0$ is the \emph{number of counters}, 
\item 
 $\leftend \notin \Sigma$ is the \emph{left} and
 $\rightend\notin \Sigma$ is the \emph{right endmarker},
\item
  $q_0 \in Q$ is the \emph{initial state}, 
\item 
 $F\subseteq Q$ is the set of \emph{accepting states}, and 
\item
  $\delta\colon Q \times(\Sigma\cup\{\leftend, \rightend\})\times \{+,\empt\}^k \to 
       Q\times\{0,1\}\times \{-1,0,1\}^k$ 
   is the partial transition function that dependent on the current state, 
   the current input symbol, and the current statuses of the counters ($+$
   indicates a positive value and $\empt$ a zero). The transition function determines the successor
   state, 
the input head movement
       ($0$ means to keep the head on the current square, and 
       $1$ means to move one square to the right), and the operations on the
       counters ($-1$ means to decrease, $+1$ to increase, and $0$ to keep the
       current value).
\end{enumerate}
\end{definition}

It is understood that the head of the input tape never moves beyond the
endmarkers and that a counter value zero is never decreased. 

A \emph{configuration} of a $\cak{k}$  
$M=\langle Q, \Sigma,k,\leftend,\rightend, \delta, q_0, F\rangle$
is a $(k+3)$-tuple $(q,w,h, c_1, c_2,\dots, c_k)$, where $q\in Q$ is the current
state, $w\in \Sigma^*$ is the input, $h\in \{0,1,\dots,|w|+1\}$ is the current 
head position on the input tape, and
$c_i\geq 0$ is the current value of counter $i$ , $1\leq i\leq k$.
The \emph{initial configuration}
for input $w$ is set to $(q_0, w, 0, 0, \dots, 0)$.
During the course of its computation,~$M$ runs through a sequence 
of configurations. One step from a configuration to its successor 
configuration is denoted by~$\vdash_M$.

A $\cak{k}$ \emph{halts} if the transition function is undefined for the
current configuration (we do not require that the head has to be placed on the
right endmarker in order to have a definition consistent with two-way devices).  
An input word $w$ is \emph{accepted} if
the machine halts at some time in an accepting state, otherwise it is
\emph{rejected}. The \emph{language accepted} by $M$ is
$L(M)=\{\, w\in \Sigma^* \mid w \text{ is accepted by } M\,\}$.

Now we turn to \emph{reversible} counter automata.
Basically, reversibility is meant with respect to the possibility of 
stepping the computation back and forth. 
So, the automata have also to be
backward deterministic.
That is, any configuration must have at most one 
predecessor which, in addition, is computable by a counter automaton.
In particular for the read-only input tape, the 
machines reread the input symbol which they have been read in a preceding 
forward computation step. 
Therefore, for reverse computation steps of \emph{one-way}
machines the head of the input tape is either moved to the \emph{left} or 
stays stationary. One can imagine that in a forward step, first the input
symbol is read and then the input head is moved to its new position, whereas
in a backward step, first the input head is moved to its new position and then 
the input symbol is read.
So, a $\cak{k}$ $M$ is said to be \emph{reversible} ($\revcak{k}$) if  
and only if there exists a \emph{reverse transition function}
$\invdelta\colon Q \times(\Sigma\cup\{\leftend, \rightend\})\times \{+,\empt\}^k \to 
       Q\times\{0,-1\}\times \{-1,0,1\}^k$ 
inducing a relation~$\invvdashm$ from a
configuration to its \emph{predecessor configuration}, so that  
\begin{multline*}
(q',w,h', c'_1, c'_2,\dots, c'_k) \invvdashm (q,w,h, c_1, c_2,\dots, c_k)
\text{ if and only if }\\
(q,w,h, c_1, c_2,\dots, c_k)\vdash_M (q',w,h', c'_1, c'_2,\dots, c'_k).
\end{multline*}

It is well known that general one-way two-counter automata 
are computational universal, that is, they can
simulate Turing machines~\cite{Minsky:1961:rupptotttm}. So, in the sequel we also
consider restricted variants. More precisely, we consider time limits for
accepting computations. 
Let $t\colon\mathbb{N}\to\mathbb{N}$ be a
function. A $\cak{k}$ $M$ is said to be \emph{$t$-time-bounded} or of
\emph{time complexity} $t$ if and only if it halts on every input $w\in L(M)$ 
after at most $t(|w|)$ time steps. A particular time bound is \emph{real
  time}, that is, the smallest time at which the counter automaton can
read the input entirely (including the right endmarker). So, here real-time
is defined to be $t(n)=n+2$. 
A $\cak{k}$ is said to be
\emph{quasi real time} if there is a constant that bounds the number of
consecutive stationary moves in all accepting computations.

The family of all languages which can be accepted by some device $\mathsf{X}$ with time
complexity $t$ is denoted by $\lfam_t(\mathsf{X})$.

To clarify our notion we continue with examples.

\begin{example}\label{exa:cfl}
The deterministic context-free language 
$\{\,w \in \{a,b\}^*\mid |w|_a=|w|_b\,\}$
is accepted by the real-time $\revcak{1}$ 
$
M=\langle \{q_0,q_1,q_a,q_b,q_f\},\{a,b\},1, \leftend,\rightend, \delta,q_0,\{q_f\}
\rangle$ where the transition functions $\delta$ and $\invdelta$
are as follows.
\begin{center}
\renewcommand{\arraystretch}{1.1}
\begin{tabular}[t]{rccc}
\hline
\multicolumn{4}{c}{$\revcak{1}$ forward}\\
\hline
(1) &  $\delta(q_0,\leftend,\bot)$ &=& $(q_1, 1, 0)$\\
(2) &  $\delta(q_1,a,\bot)$ &=& $(q_a, 1, 0)$\\
(3) &  $\delta(q_1,b,\bot)$ &=& $(q_b, 1, 0)$\\
(4) &  $\delta(q_1,\rightend,\bot)$ &=& $(q_f, 0, 0)$\\
(5) &  $\delta(q_a,a,\bot)$ &=& $(q_a, 1, +1)$\\
(6) &  $\delta(q_a,b,\bot)$ &=& $(q_1, 1, 0)$\\
(7) &  $\delta(q_a,a,+)$ &=& $(q_a, 1, +1)$\\
(8) &  $\delta(q_a,b,+)$ &=& $(q_a, 1, -1)$\\
(9) &  $\delta(q_b,a,\bot)$ &=& $(q_1, 1, 0)$\\
(10) &  $\delta(q_b,b,\bot)$ &=& $(q_b, 1, +1)$\\
(11) &  $\delta(q_b,a,+)$ &=& $(q_b, 1, -1)$\\
(12) &  $\delta(q_b,b,+)$ &=& $(q_b, 1, +1)$\\
\hline
\end{tabular}
\quad
\begin{tabular}[t]{rccc}
\hline
\multicolumn{4}{c}{$\revcak{1}$ backward}\\
\hline
(1) &  $\invdelta(q_1,\leftend,\bot)$ &=& $(q_0, -1, 0)$\\
(2) &  $\invdelta(q_1,a,\bot)$ &=& $(q_b, -1, 0)$\\
(3) &  $\invdelta(q_1,b,\bot)$ &=& $(q_a, -1, 0)$\\
(4) &  $\invdelta(q_f,\rightend,\bot)$ &=& $(q_1, 0, 0)$\\
(5) &  $\invdelta(q_a,a,\bot)$ &=& $(q_1,-1,0)$\\
(6) &  $\invdelta(q_a,b,\bot)$ &=& $(q_a,-1,+1)$\\
(7) &  $\invdelta(q_a,a,+)$ &=& $(q_a,-1,-1)$\\
(8) &  $\invdelta(q_a,b,+)$ &=& $(q_a,-1,+1)$\\
(9) &  $\invdelta(q_b,a,\bot)$ &=& $(q_b,-1,+1)$\\
(10) &  $\invdelta(q_b,b,\bot)$ &=& $(q_1,-1,0)$\\
(11) &  $\invdelta(q_b,a,+)$ &=& $(q_b,-1,+1)$\\
(12) &  $\invdelta(q_b,b,+)$ &=& $(q_b,-1,-1)$\\
(13) &  $\invdelta(q_1,\rightend,\bot)$ &=& $(q_1,-1,0)$\\
\hline
\end{tabular}
\end{center}
The basic idea of the construction is as follows. We use the counter for storing the
difference between the number of $a$'s and $b$'s in the input. However, to enable
the deterministic backward computation, the difference one is remembered in the
states and the counter is only used to store larger differences. Hence, the state 
$q_a$ indicates that there are more $a$'s than $b$'s in the input read so far 
and $q_b$ denotes the opposite. Now, the computation is started with transition (1)
which moves from the left endmarker to the first symbol and enters state $q_1$ 
that indicates that the number of $a$'s and $b$'s currently read is equal. 
Then, transitions (2) and (3) are used to count the difference one. 
Transitions (5), (7) and (10), (12) increase the difference by one and
transitions (8) and (11) decrease the difference by one.
Finally, if the difference is one, then transitions (6) and (9) can be used to
decrease the difference to zero, to enter state $q_1$, and to enter an
accepting state when reading the right endmarker with transition (4).
For the backward computation we just have to do the opposite by switching the
roles of $a$ and $b$. For example, transitions (5) and (7) of $\delta$ increase 
the difference by one when an $a$ is read and  
there have been more~$a$'s than $b$'s read so far. 
This difference is later decreased by one with transitions~(6) and (8)
when a~$b$ is read. Thus, for $\invdelta$ we have to increase the difference when
reading a $b$ (transitions (6) and (8)) and to decrease the difference when
reading an~$a$ (transitions~(5) and (7)). 
The transitions (2)--(3) and \mbox{(9)--(12)} can analogously be translated. 
To translate the transitions (1) and (4) concerning the endmarkers is 
straightforward.
It follows immediately from the transition function that $M$ moves its head in
any but the last computation step. So, it takes at most $n+2$ steps, that is,
it
works in real time.
\eoe
\end{example}

\begin{example}\label{exa:non-cfl}
The non-context-free language 
$\{\,w \in \{a,b,c\}^*\mid |w|_a=|w|_b=|w|_c\,\}$
can be accepted by a real-time $\revcak{2}$.
The basic idea is to  
implement the construction described in Example\ref{exa:cfl} twice,
namely, one counter is used to check whether the number of $a$'s is equal to the
number of $b$'s and the other counter is used to check whether the number of $a$'s is
equal to the number of $c$'s. More precisely, the construction can be 
realized using the Cartesian product of the construction from Example\ref{exa:cfl},
and one component of the state set and one counter suffices to check the
difference between $a$'s and $b$'s ($c$'s are ignored) and 
$a$'s and $c$'s ($b$'s are ignored), respectively.
The input is accepted if in both components the state $q_1$ is reached
when reading the right endmarker. In this case, the numbers of $a$'s and $b$'s are equal
as well as the numbers of $a$'s and $c$'s. Hence, the numbers of $b$'s and $c$'s
are equal as well.
Since the computation in each component is
reversible, the overall computation is also reversible.

This idea can straightforwardly be generalized to show that the language
$$\{\,w \in \{a_1,a_2,\ldots,a_k\}^*\mid |w|_{a_1}=|w|_{a_2}=\cdots=|w|_{a_k}\,\}$$
 for $k \ge 2$ and an alphabet $\{a_1,a_2,\ldots,a_k\}$ of $k$ symbols
 can be accepted by a real-time $\revcak{k-1}$.
 \eoe
\end{example}

\section{Computational Capacity of Reversible Counter Automata}\label{sec:compcap}

Here, we consider the computational capacities of reversible counter automata
and compare it with the general variants.
First, we are interested in the role played by stationary moves in 
quasi real-time computations. In order to settle this role, we first
deal with a more general issue. For, not necessarily reversible, counter
automata the restriction to be able to add or subtract only $1$ per step
to or from the counters is not a limitation of the computational capacity. Clearly,
any counter automaton that may add or subtract an arbitrary number to or from
the counters in a single step can be simulated by a sequence of stationary moves that increment
or decrement the counters by only $1$ per step. 
However, the simulation can be done without loss of time. It has been
mentioned in~\cite{fischer:1968:cmacl} without details. Here, we will
show that the construction can be done such that reversibility is preserved.

\newpage

\begin{lemma}\label{lem:add-c}
Let $k, c \geq 0$ be integers. For every $\revcak{k}$ that obeys some time complexity
$t(n)$ and that has the ability to alter the value of each counter independently
by any integer between $-c$ and $c$ in a single step, 
an equivalent ordinary $\revcak{k}$ obeying the time complexity
$t(n)$ can effectively be constructed.
\end{lemma}

\begin{proof}
Let $M=\langle Q, \Sigma,k,\leftend,\rightend, \delta, q_0, F\rangle$
be a $\revcak{k}$ that has the ability to alter the value of each counter 
independently by any integer between $-c$ and $c$ in a single step.
The basic idea of the construction of an equivalent ordinary
$\revcak{k}$ $M'=\langle Q', \Sigma,k,\leftend,\rightend, \delta', q'_0,
F'\rangle$ is as follows. A counter value~$x$ of $M$ is represented by
the counter value $\lfloor \frac{x}{c}\rfloor$ and a state component
that stores $x\bmod c$. To this end, we set
$Q'= Q\times \{0,1,\dots, c-1\}^k$, $q'_0=(q_0,(0,0,\dots,0))$,
and $F'= F\times \{0,1,\dots, c-1\}^k$.
The transition function $\delta'$ has to be constructed, in particular, such
that it is reversible.

For $q \in Q$, $m_1, m_2,\dots, m_k\in  \{0,1,\dots, c-1\}$, 
$a\in(\Sigma\cup\{\leftend, \rightend\})$, $d_1, d_2,\dots, d_k\in\{+,\empt\}$, we define
\begin{equation}\label{eq:deltan}
\delta'((q,(m_1,m_2,\dots, m_k)), a, d_1,d_2,\dots, d_k)
= 
((q', (m_1', m_2',\dots, m_k')), s, b_1, b_2,\dots, b_k)
\end{equation}
if and only if
\begin{equation}\label{eq:deltao}
\delta(q,a,\hat{d}_1, \hat{d}_2,\dots, \hat{d}_k)
=
(q',s, \hat{b}_1, \hat{b}_2,\dots,\hat{b}_k)
\end{equation}
where
$\hat{d}_i = \empt$ if $m_i=0$ and $d_i=\empt$, and $\hat{d}_i = +$ otherwise,
and
$$
(m_i',b_i)=\begin{cases}
(m_i+\hat{b}_i,0) & \text{ if } 0\leq m_i+\hat{b}_i \leq c-1\\
(m_i+\hat{b}_i+c,-1) & \text{ if } m_i+\hat{b}_i <0\\
(m_i+\hat{b}_i-c,1) & \text{ if } m_i+\hat{b}_i > c-1\\
\end{cases},
$$
for $1\leq i\leq k$. Note that $-1\leq b_i\leq 1$ and $-c\leq \hat{b}_i\leq
c$.

\noindent
Immediately, from the construction it follows that
$(q_0, w, 0,0,\dots,0)\vdash_M^* (q, w, h, c_1,c_2,\dots, c_k)$
if and only if 
$((q_0,(0,0,\dots,0)), w, 0,0,\dots,0)\vdash_{M'}^* 
       ((q,(c_1\bmod c,\dots, c_k\bmod c)), w, h, 
              \lfloor\frac{c_1}{c}\rfloor ,\dots,
              \lfloor\frac{c_k}{c}\rfloor)$.
So, we conclude that $L(M)=L(M')$ and that $M$ and $M'$ share the same time
complexity.
It remains to be shown that~$M'$ is reversible.

Since $M$ is reversible, the transition from Equation~(\ref{eq:deltao}) can be
reversed, say by
$$
\invdelta(q',a,\tilde{d}_1, \tilde{d}_2,\dots, \tilde{d}_k)=(q,s, -\hat{b}_1,
-\hat{b}_2,\dots,-\hat{b}_k).
$$
Then, we construct
\begin{equation}\label{eq:deltanr}
\invdeltap((q',(m'_1,m'_2,\dots, m'_k)), a, d'_1,d'_2,\dots, d'_k)
= 
((q, (m''_1, m''_2,\dots, m''_k)), s, b'_1, b'_2,\dots, b'_k)
\end{equation}
where
$\tilde{d}_i = \empt$ if $m'_i=0$ and $d'_i=\empt$, and $\tilde{d}_i = +$
otherwise, and
$$
(m''_i,b'_i)=\begin{cases}
(m'_i-\hat{b}_i,0) & \text{ if } 0\leq m'_i-\hat{b}_i \leq c-1\\
(m'_i-\hat{b}_i+c,-1) & \text{ if } m'_i-\hat{b}_i <0\\
(m'_i-\hat{b}_i-c,1) & \text{ if } m'_i-\hat{b}_i > c-1\\
\end{cases},
$$
for $1\leq i\leq k$. 

In order to show that Equation~(\ref{eq:deltanr}) reverses Equation~(\ref{eq:deltan}) we distinguish
the three cases of the construction of $(m'_i,b_i)$.
\begin{description}
\item[Case $(m'_i,b_i)=(m_i+\hat{b}_i,0):$]
  Here we know $b_i=0$ and $0\leq m_i+\hat{b}_i \leq c-1$.
  Since $m'_i-\hat{b}_i = m_i+\hat{b}_i-\hat{b}_i=m_i\in\{0,1,\dots,c-1\}$ we
  derive that $(m''_i,b'_i)$ has been set to $(m'_i-\hat{b}_i,0)$. Therefore, we
  conclude $m''_i=m'_i-\hat{b}_i=m_i$ and $b'_i=0$. So, Equation~(\ref{eq:deltanr})
  reverses Equation~(\ref{eq:deltan}) in this case.
\item[Case $(m'_i,b_i)=(m_i+\hat{b}_i+c,-1):$]
  The condition for this case is $m_i+\hat{b}_i <0$ and we know $b_i=-1$.
Since $m'_i-\hat{b}_i = m_i+\hat{b}_i+c-\hat{b}_i=m_i+c > c-1$ we
  derive that $(m''_i,b'_i)$ has been set to $(m'_i-\hat{b}_i-c,1)$. 
Therefore, we
  conclude $m''_i=m'_i-\hat{b}_i-c=m_i$ and $b'_i=1$. So, Equation~(\ref{eq:deltanr}) 
  reverses Equation~(\ref{eq:deltan}) also in this case.

\item[Case $(m'_i,b_i)=(m_i+\hat{b}_i-c,1):$]
  Here we have $b_i=1$ and $m_i+\hat{b}_i >c-1$.
  Since $m'_i-\hat{b}_i = m_i+\hat{b}_i-c-\hat{b}_i=m_i-c < 0$ we
  derive that $(m''_i,b'_i)$ has been set to $(m'_i-\hat{b}_i+c,-1)$. 
  Therefore, we
  conclude $m''_i=m'_i-\hat{b}_i+c=m_i$ and $b'_i=-1$. So, Equation~(\ref{eq:deltanr})
  reverses Equation~(\ref{eq:deltan}) in this case, as well.
\end{description}
From the three cases we derive that $M'$ is reversible.
\end{proof}

The next step is to use Lemma~\ref{lem:add-c} to show that
quasi real-time computations can be sped-up to real time.

\begin{theorem}\label{theo:dcak-quasirealtime}
Let $k \geq 0$ be an integer. 
For every quasi real-time $\revcak{k}$ an equivalent real-time $\revcak{k}$ can
effectively be constructed.
\end{theorem}

\begin{proof}
Let $M$
be a quasi real-time $\revcak{k}$ that never performs more than $\ell\geq 0$
stationary moves consecutively. Clearly, if $\ell=0$ then $M$ does not perform
a stationary move at all and, thus, works in real time. So, we consider
$\ell\geq 1$ in the rest of the proof.

The first step is the construction of an equivalent 
$\revcak{k}$~$M'=\langle Q', \Sigma,k,\leftend,\rightend, \delta', q_0',
F'\rangle$
as in the proof of Lemma~\ref{lem:add-c}, 
that may alter the value of each counter independently
by any integer between $-(\ell+1)$ and $\ell+1$ in a single step.
We let $M'$ simulate $M$ step by step and derive that $M'$ will not 
use its extended abilities, since it will change its counter values 
by at most one in each move. Moreover, $M'$ works still in quasi real-time.

The next step is to speed-up $M'$ to real-time. To this end,
we modify $M'$ to an equivalent real-time 
$\revcak{k}$~$M''=\langle Q', \Sigma,k,\leftend,\rightend, \delta'', q_0',
F'\rangle$. Note that due to its construction, $M'$ knows in each step
whether the represented value of its counter $c_i$ is at least $\ell+1$. If it is strictly
less than  $\ell+1$ then $M'$ knows its exact value, $1\leq i\leq k$.
The purpose of the extended abilities of $M'$ is
to simulate at once a possibly empty sequence of
stationary moves and a possibly subsequent non-stationary step by $M''$.
So,~$M''$ will be able to detect whether its counters can get
empty within the next $\ell+1$ steps. 

By the construction in the proof
of Lemma~\ref{lem:add-c}, a counter value of $M''$ is represented by the sum
of the counter value itself times $\ell+1$ and a state component that stores a number 
from~$\{0,1,\dots,\ell\}$.
So, we have 
$Q'= Q\times \{0,1,\dots, \ell\}^k$, $q'_0=(q_0,(0,0,\dots,0))$,
and $F'= F\times \{0,1,\dots, \ell\}^k$.
Next, the transition function $\delta''$ has to be constructed, in particular, such
that the reversibility of $M'$ is preserved.

\begin{sloppypar}
Given $q\in Q'$, $a\in(\Sigma\cup\{\leftend, \rightend\})$, and
$d_1, d_2,\dots, d_k\in\{+,\empt\}$, 
the transition 
$\delta''(q, a, d_1,d_2,\dots, d_k)$
is defined by the computation $\gamma_1 \vdash_{M'} \gamma_2 \vdash_{M'} \cdots \vdash_{M'} \gamma_n$
of~$M'$ starting on 
$$
\gamma_1=
(q,x,h, c_1, c_2,\dots, c_k)
$$
where $(x,h)=(\lambda,0)$ if $a=\leftend$, and $(x,h)=(a,1)$ otherwise, as
well as $c_i=1$ if $d_i=+$, and $c_i=0$ if $d_i=\empt$, $1\leq i\leq k$.
The computation starts with a possibly empty sequence of stationary moves
on input symbol $x$, followed by a non-stationary move on~$x$. Since $M'$ 
works in quasi real-time this takes at most~\mbox{$\ell+1$} steps.
\end{sloppypar}

First, we construct $\delta''$ for the cases where the computation does not 
halt before step $\ell+1$.
Let $\gamma_n$ be the configuration reached after the 
non-stationary move. Now, assume the state of $\gamma_n$ is $q'$, 
and~$M'$ has altered the represented value of counter $c_i$ by some $-(\ell+1)\leq j_i\leq \ell+1$
(recall that $M'$ changes its represented counter values by at most one in each move).
Then the transition $\delta''(q, a, d_1,d_2,\dots, d_k)$
of~$M''$ to be defined yields
$(q', 1, j_1, j_2,\dots, j_k)$.

Second, assume that the computation halts in configuration $\gamma_n$
before step $\ell+1$, that is, before performing the non-stationary move.
Then the transition $\delta''(q, a, d_1,d_2,\dots, d_k)$
of~$M''$ to be defined yields
$(q', 0, j_1, j_2,\dots, j_k)$.

So, from the construction we obtain that, given an input $w$, 
the computation of $M'$ is unambiguously
split into sequences of steps each of which is performed by $M''$
at once. If $M'$ accepts, so does $M''$ also in cases where
the input is accepted after some stationary moves at the end of the
computation. Conversely, every step of~$M''$ corresponds
to a sequence of steps of $M'$. So, we have 
$L(M') = L(M'')$. Moreover, $M''$ works in real-time.

Finally, the reversibility of $M''$ follows by the reversibility
of $M'$. Since $M'$ is reversible, the computation 
$c_n \invvdashmp c_{n-1}\invvdashmp\cdots \invvdashmp c_1$ is unique.
There remains only one point. In the backwards computation, first
the non-stationary move is simulated followed by some stationary moves.
While this causes no trouble in general, we have to argue that
the computation does not go before the initial configuration
with stationary moves. However, since a loop with stationary
moves on the left endmarker that runs from the initial configuration
to the initial configuration would imply that the accepted language
is empty, we safely may assume that it does not exist. So, 
any stationary transition on the left endmarker that leads to the
initial configuration can safely be removed from $M''$. Note, that
these transitions can be identified by the construction of $M''$.
We conclude that $M''$ is reversible.
\end{proof}

So, the family of languages accepted by quasi real-time $\revcak{k}$
equals the family of languages accepted by real-time $\revcak{k}$.

Next, we turn to the question of whether the property of being reversible
causes weaker computing capabilities for counter automata at all.
It is known that reversible two-counter automata 
that do not have an input tape but receive their inputs suitably
encoded into their counters can simulate Turing machines~\cite{morita:1996:urtcm}.
So, we have to consider counter automata working within time bounds.
Though reversible counter automata are able to accept even non-context-free
languages in real time (Example~\ref{exa:non-cfl}), their reversibility has a drastic
impact on their computational capacities for certain languages.
We will show that there is a regular language not accepted by
any reversible counter automaton (with an arbitrary number of counters) 
with the super-polynomial time complexity $2^{o(n)}$.
To this end, we will use Kolmogorov complexity 
and incompressibility arguments.
General information on this technique
can be found, for example, in the textbook~\cite[Ch.~7]{li:1993:itkca:book}.
Let $w\in \{0, 1\}^*$ be an arbitrary binary string. The Kolmogorov 
complexity $C(w)$ of $w$ is defined to be the minimal size of a 
binary program (Turing machine) describing~$w$. The following key component 
for using the incompressibility
method is well known: there are binary strings~$w$ 
of \emph{any} length such that $|w| \le C(w)$.

\begin{theorem}\label{theo:non-reg}
Let $k\geq 0$ be an integer. There exists a regular language that is not
accepted by any $2^{o(n)}$-time $\revcak{k}$.
\end{theorem}

\begin{proof}
We consider the regular language $L=((aa + a)(bb + b))^*(aa + a + \lambda)$
as witness.
Assume in contrast to the assertion that $L$ is accepted by some 
$\revcak{k}$ $M=\langle Q, \Sigma,k,\leftend,\rightend, \delta, q_0, F\rangle$
in some time $t(n)=2^{o(n)}$.

We choose a word $w \in \{0,1\}^+$ long enough 
such that $C(w)\ge |w|$.
Now, $w$ is encoded as follows.
From left to right the digits are represented 
alternating by $a$'s and $b$'s such that a $0$
is represented by a single letter and a $1$ by
a double letter. For example, the word
$010110$ is encoded as $abbabbaab$. Let~$\varphi(w)$
denote the code of $w$. We have $\varphi(w)\in L$.
Next, we consider the accepting computation on $\varphi(w)$
and show that $w$ can be compressed.

Since $M$ accepts in time $2^{o(n)}$, the maximum number
stored in some counter of $M$ in the accepting computation on $\varphi(w)$
is bounded from above by $2^{o(n)}$. Therefore, 
omitting the second component, each configuration 
$(q,\varphi(w),h, c_1, c_2,\dots, c_k)$
of $M$ can be encoded with 
$$\lceil \log(|Q|)\rceil + \lceil\log(|\varphi(w)|+2)\rceil + k\cdot
o(|\varphi(w)|) 
= o(|\varphi(w)|) = o(|w|)
$$ 
bits. 

Knowing $M$, the length of $\varphi(w)$, and the accepting configuration on 
$\varphi(w)$ without the second
component, $w$ can be reconstructed as follows. For each candidate 
string~$x$ of length~$|\varphi(w)|$, the $\revcak{k}$~$M$ is simulated.

We claim that if the simulation accepts in the accepting configuration 
of $\varphi(w)$ then we have $x=\varphi(w)$ and, thus, decoding $\varphi(w)$ yields $w$.

\begin{sloppypar}
In order to show the claim, assume that $x\neq \varphi(w)$.
Then the computation is run backwards as long
as the suffixes of $\varphi(w)$ and $x$ are identical, thus, reaching some
configurations 
$(q, uzv, |u|+2, c_1, c_2,\dots, c_k)$ and
$(q, u'z'v, |u|+2, c_1, c_2,\dots, c_k)$
with $uzv=\varphi(w)$ and
$u'z'v=x$, $z,z' \in \{a,b\}$, $|u|=|u'|$, and $z\ne z'$. 
We may safely assume that $z=a$ and $z'=b$.
Since
$uzbb=uabb$ belongs to $L$ the computation continuing in
\mbox{$(q, uzbb, |u|+2, c_1, c_2,\dots, c_k)$} ends accepting.
But then the computation continuing in
\mbox{$(q, u'z'bb, |u|+2, c_1, c_2,\dots, c_k)$} is accepting as well.
However the input $u'z'bb=u'bbb$ has to be rejected
since it ends with three $b$'s. This contradiction shows the
claim.
\end{sloppypar}

We conclude that the Kolmogorov complexity of $w$
is $C(w) = o(|w|)+ \lceil \log(|\varphi(w)|)\rceil + \ell = o(|w|)$,
for a positive constant $\ell$ which gives the size of $M$ and the program
that reconstructs~$w$. So, we have $C(w) < |w|$, for $w$ long enough. 
This is a contradiction since $w$ has been chosen such that $C(w) \ge |w|$.
The contradiction shows that $L$ is not accepted by~$M$.
\end{proof}

Since even $\cak{0}$, which are essentially $\dfa$s, can accept all
regular languages, we have separated the computational capacity
of $\cak{k}$ and $\revcak{k}$ for all $k\geq 1$ if they obey the same time
complexity.

\begin{theorem}\label{theo:reg-sep}
Let $k\geq 0$ be an integer. 
The family of languages accepted by $\revcak{k}$
in at most $2^{o(n)}$-time is strictly
included in the family of languages 
accepted by $\cak{k}$ in at most $2^{o(n)}$ time.
\end{theorem}

Next, we turn to the impact of the number of counters to the computational
capacities of reversible counter automata. Infinite and strict counter
hierarchies for general counter automata working in real-time 
are known for a long time~\cite{fischer:1968:cmacl,laing:1967:racce}.
Generalizations to polynomial and exponential time complexities
have been obtained in~\cite{petersen:2011:sbtbcm}.
However, the hierarchy results~\cite{petersen:2011:sbtbcm}
rely on a stronger acceptance condition. This stronger condition
requires that the computations on \emph{all} inputs have to
respect the time complexity. In particular, this stronger
condition weakens the non-acceptance results. On passing, here
we obtain the known hierarchies also for ordinary counter automata
even for the weaker acceptance condition as used overall in this paper.
In our definition only accepting computations have to obey the
time complexity. To obtain the results, we use once more Kolmogorov
arguments.

Our next counter hierarchy concerns reversible and general
counter automata working in some exponential time. 
In order to define languages that serve as witnesses,
let \mbox{$\varphi\colon \{a,b,\bar{a},\bar{b}\}^*\to \{0,1\}^*$} be the homomorphism
defined through $\varphi(a)=\varphi(\bar{a})=0$ and $\varphi(b)=\varphi(\bar{b})=1$.
Next, we consider all words~$w$ over the alphabet $\{a,b,\bar{a},\bar{b}\}$ as binary
numbers~$\varphi(w)$. The integer represented by $\varphi(w)$ is denoted by
$\eta(\varphi(w))$. 
Let $k\geq 2$ and $j\geq 1$ be integers and $\varphi(w)=z_1z_2\cdots z_{j\cdot
  k} \in\{0,1\}^{j\cdot k}$. 
Then, for all $1\leq i\leq k$, we consider the scattered factors 
\mbox{$v^{(k)}_i(\varphi(w))= z_iz_{k+i}z_{2k+i}\cdots z_{(j-1)k+i}$} of~$\varphi(w)$.
Now, for all $k\geq 2$, we define the language
\begin{multline*}
L_k = \{\, uz_1\dollar^iz_2v \mid j\geq 1, u\in\{a,b\}^{j\cdot k-1}, z_1,z_2\in\{\bar{a},\bar{b}\},
1\leq i\leq k,\\ v\in\{a,b\}^*, \eta(v^{(k)}_i(\varphi(uz_1)))=\eta(\varphi(z_2v)^R)\geq 1\,\}.
\end{multline*}

\begin{proposition}\label{prop:hier-positive}
Let $k\geq 2$ be an integer. The language $L_k$ is accepted by a
$\revcak{k+1}$ with time complexity $O(2^{\frac{n}{k}})$.
\end{proposition}

In order to prove that $k$ counters are not enough to accept $L_k$ in time
$O(2^{\frac{n}{k}})$, we use again Kolmogorov complexity 
and incompressibility arguments. In particular, $k$ counters are not enough
even for not necessarily reversible counter automata.

\begin{proposition}\label{prop:hier-negative}
Let $k\geq 2$ be an integer. The language $L_k$ is not accepted 
by any $\cak{k}$ with time complexity $O(2^{\frac{n}{k}})$.
\end{proposition}

\begin{proof}
Assume contrarily that $L_k$ is accepted by some 
$\cak{k}$ $M=\langle Q, \Sigma,k,\leftend,\rightend, \delta, q_0, F\rangle$
with time complexity $O(2^{\frac{n}{k}})$.

We choose some integer $j\geq 1$ large enough
and a word $u'\in \{a,b\}^{j\cdot k}$
such that $C(u')\ge |u'|$.
Next, we consider the computation of $M$
on the prefix $uz_1$, where $uz_1$ is essentially
$u'$ but with the last symbol barred. In particular,
we consider the configuration reached exactly after
$M$ has moved its input head from $z_1$ to the right.

Set $\ell=\frac{n}{k}$.
Since $M$ accepts in time $O(2^{\frac{n}{k}})$
the maximum number stored in some counter of $M$ 
in the configuration is at most $c\cdot 2^\ell$, for some
constant $c\geq 0$.
Moreover, in~\cite{petersen:2011:sbtbcm} it is shown that in an accepting computation
if a counter machine has $|Q|$ states then immediately after reading a
prefix $u'$ of its input the value of at least one counter 
is bounded from above by $(|Q|+1)\cdot |u'|$. Since this could be each
of the $k$ counters, the configuration in question without its second and
third component
can be encoded with
\begin{eqnarray*}
\lceil \log(|Q|\cdot (c\cdot 2^\ell)^{k-1} \cdot k\cdot (|Q|+1)\cdot
|u'|)\rceil 
&\leq& c'+ \ell(k-1) + \log(|u'|)\\
&=&  c' + \frac{n}{k}(k-1) + \log(|u'|)
\end{eqnarray*}
bits, for some constant $c'\geq 0$.
Since $n \leq j \cdot k + j + k$, we have $\frac{n-k}{k+1}k \leq j\cdot k =
|u'|$ and conclude
\begin{eqnarray*}
c' + \frac{n}{k}(k-1) + \log(|u'|) 
&\leq& c'' + \left(\frac{|u'|(k+1)}{k}+k\right)\frac{k-1}{k} + \log(|u'|)\\
&=& c'' + \frac{|u'|(k+1)(k-1)}{k^2}+ \frac{(k-1)k}{k} + \log(|u'|)\\
&=& c'' +|u'|- \frac{|u'|}{k^2} + k-1 + \log(|u'|)\\
&<& |u'|
\end{eqnarray*}

Knowing $M$, the length of $u'$, and the configuration in question  
without its second and
third component, $u'$ can be reconstructed as follows.
In order to reconstruct $u'$ it is sufficient to reconstruct
the $k$ scattered factors $v^{(k)}_i(\varphi(u'))$. 
For each of these factors, we test all non-empty candidates
from $\{a,b\}^{\leq \frac{|u'|}{k}}$. A candidate $v$ is tested
by replacing its first symbol by the barred version, preceding it by
$\dollar^i$, and feeding it to $M$. Now it is sufficient to
simulate $M$ on this input starting in the encoded configuration.
The simulation can be made halting due to the upper bound of the time
complexity. So, if $M$ accepts then $v^R$ preceded with leading $a$'s
to obtain the length $\frac{|u'|}{k}$ gives the scattered factor $v^{(k)}_i(\varphi(u'))$.
Moreover, there is an accepted candidate for each $i$. Therefore, 
the reconstruction terminates.

We conclude that the Kolmogorov complexity of $u'$
is bounded from above by a positive constant 
which gives the size of $M$ and the program that reconstructs~$u'$
plus $\log(|u'|)$ plus the size of the encoding of the configuration. In total this
is strictly less than $|u'|$.
However, this is a contradiction since $u'$ has been chosen such 
that $C(u') \ge |u'|$.
The contradiction shows that $L_k$ is not accepted by~$M$
with time complexity $O(2^{\frac{n}{k}})$.
\end{proof}

So, we have a strict counter hierarchy for reversible as well as irreversible
counter automata obeying the exponential time complexity~$O(2^{\frac{n}{k}})$.
However, to compare reversible counter automata with some $k+1$ counters
to irreversible counter automata with $k$ counters, we cannot utilize
the regular language provided by Theorem~\ref{theo:non-reg}.
This language is not accepted by any $2^{o(n)}$-time $\revcak{k}$.
So, counter hierarchies for polynomial time complexities can be considered.
Such hierarchies are known to exist for irreversible counter automata
that have to respect the time complexity on all inputs, that is, also
on inputs \emph{not accepted}~\cite{petersen:2011:sbtbcm}.
Here we just mention that this result can be improved by Kolmogorov
arguments, such that only accepting computations have
to respect the time complexity. 
The witness languages are modifications of the languages $L_k$ used to
show Proposition~\ref{prop:hier-positive} and 
Proposition~\ref{prop:hier-negative}.
The relations between language families are depicted in
Figure~\ref{fig:pic-incl}.

\begin{figure}[!ht]
  \centering
  \includegraphics[scale=.7]{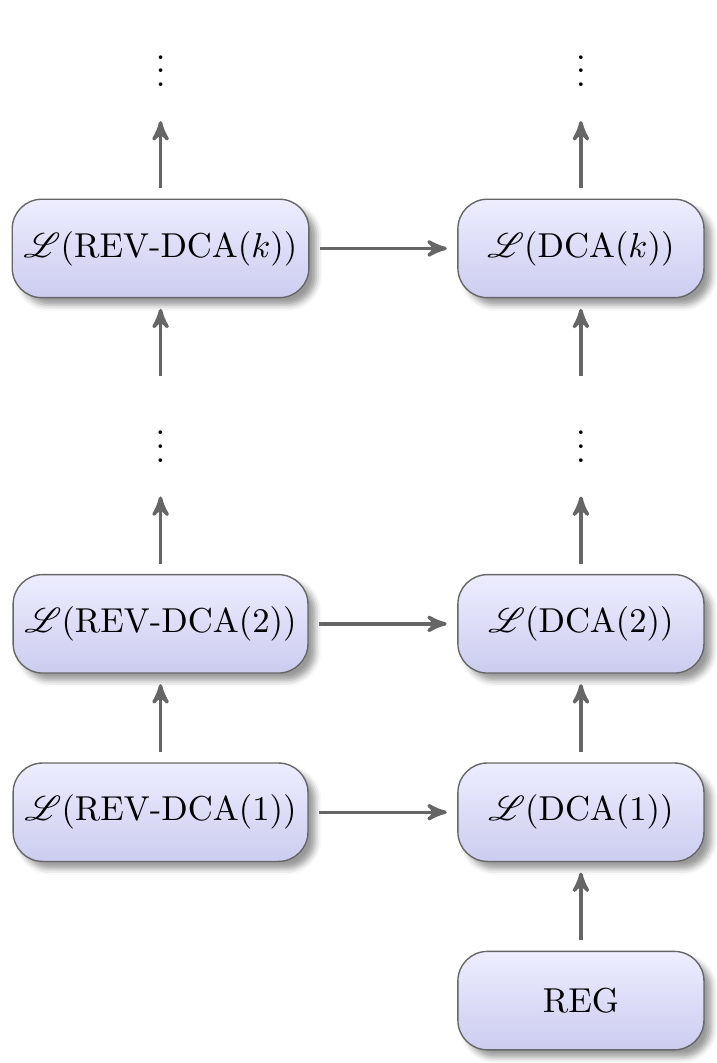} 
  \caption{Relationships between language families induced by
counter automata. An arrow between families indicates a strict
inclusion. Whenever two families are not connected by a path they
are incomparable.}\label{fig:pic-incl}
\end{figure}

\section{Decidability Problems}\label{sec:deci}

In this section, we study decidability problems for real-time $\revcak{k}$ with $k \ge 2$
and it turns out at all usually studied questions such as emptiness, finiteness,
inclusion, equivalence, or regularity are not semidecidable.
We recall (see, for example, \cite{Hopcroft:1979:itatlc:book}) that a decidability
problem is \emph{semidecidable} (\emph{decidable}) if and only if the set of all
instances for which 
the answer is `yes' is recursively enumerable (recursive). 
Clearly, any decidable problem is also semidecidable, while the converse does
not generally hold.
The non-semidecidability results are shown by reduction
of the emptiness problem of Turing machines.
It is well known that emptiness for such machines is not semidecidable
(see, for example,~\cite{Hopcroft:1979:itatlc:book}).

The technique to obtain this reduction is based on the notion
of \emph{valid computations of multiplying counter machines} where we are
following an idea and the notation given in~\cite{Hartmanis:1970:wmsltpu}.
Valid computations of a multiplying counter machine are, basically, histories of 
multiplying counter machine computations which are encoded into single words.
It will be shown that the set of such suitably formatted valid
computations can be represented as the intersection of two languages
accepted by real-time $\revcak{1}$. Since it is known~\cite{Hartmanis:1970:wmsltpu}
that every language accepted by a Turing machine can be accepted,
suitably encoded, by some multiplying counter machine as well, the emptiness
of a Turing machine can be reduced to the emptiness of the intersection
of two languages accepted by real-time $\revcak{1}$s. 
{F}rom this the non-semidecidability of inclusion for real-time $\revcak{1}$s and
the non-semidecidability results for real-time $\revcak{k}$ with $k \ge 2$
can be derived. 

We start by summarizing the necessary notations given in~\cite{Hartmanis:1970:wmsltpu}.
A \emph{multiplying counter machine} is a one-register machine (the register is capable of holding
an arbitrary integer) which can multiply the content of its register by one of a
finite number of multiplicands, and branching if the resulting product is not an integer.
Formally, let $Q$ be a finite set of states with two distinguished elements
$q_0 \neq q_f$, called the initial and final state.
Let $C=\{2,3,5,7,1/2,1/3,1/5,1/7\}$ be a finite set of multiplicands. 
A transition rule is an element of $Q \times C \times Q \times Q$.
A multiplying counter machine~$M$ is a set of transition rules such that no two transition
rules have the same
first component and no transition rule of~$M$ has $q_f$ as the first component nor $q_0$ as
the last component. A configuration of $M$ is string of the form $qa^i$, where $q \in Q$. 
For each integer~$n$, we write $qa^n \vdash pa^{kn}$, if $(q,k,p,r)$ is a transition rule of $M$
and $kn$ is an integer. We write $qa^n \vdash ra^{n}$, if $(q,k,p,r)$ is a transition rule of $M$
and $kn$ is not an integer. A valid computation is a string built from a
sequence of configurations passed through during a computation halting in
state $q_f$.

\begin{theorem}[\cite{Hartmanis:1970:wmsltpu}]\label{theo:Hartmanis}
Let $L \subseteq \{1,2\}^*$ be a set of strings accepted by a Turing machine
and let each string $x_1x_2 \cdots x_n \in \{1,2\}^*$ be encoded as a natural number
$i=x_1 + 3x_2 + 3^2x_3 + \cdots + 3^{n-1}x_n$.
Then, a multiplying counter machine $M$ can effectively be constructed such that
$q_0a^{2^i} \vdash^* q_fa^j$ (for some $j \ge 0$) if and only if $i=x_1 + 3x_2 + 3^2x_3 + \cdots + 3^{n-1}x_n$
and $x_1x_2 \cdots x_n$ is in $L$.
\end{theorem}

Now, a valid computation of a multiplying counter machine~$M$ ($\valc(M)$) is defined 
in~\cite{Hartmanis:1970:wmsltpu} as a string 
$$
q_0aq_0a^2q_0a^4 \cdots q_0a^{2^{i-1}}\alpha_0\alpha_1 \cdots \alpha_n
$$
where $\alpha_0, \alpha_1, \ldots, \alpha_n$ are configurations of~$M$
such that $\alpha_0 = q_0a^{2^i}$, $\alpha_n = q_fa^{p}$ for some $p\ge 1$, 
and $\alpha_j \vdash \alpha_{j+1}$ for $0 \le j < n$. 

Our goal is to represent the set of valid computations as the intersection of two
languages where each of which is accepted by a 
real-time \emph{reversible} one-counter machine. 
Hence, we consider several modifications to the set of valid computations
to enable a reversible computation.
First, we mark the $a$ in the first block, that is, $q_0a$ is replaced by $q_0a'$.
Second, we assume that $n+i+1$ is even, that is,
every valid computation consists of an even number of configurations. 
If $n+i+1$ is odd, then we consider a new state $q_{f'}$, replace
$\alpha_n=q_f a^p$ by $\alpha_n=q_{f'} a^p$, add a configuration
$\alpha_{n+1}=q_f a^p$, and obtain an even number of configurations.

Third, we add the state of each configuration to the end of a configuration.
That is, $\alpha_j=q_ja^{n_j}$ is modified to $\alpha_jq_j=q_ja^{n_j}q_j$ and
for the initial phase of doubling $a$-blocks we modify $q_0 a^{2^p}$ ($p \ge 0$)
to $q_0' a^{2^p}q_0'$, where $q_0'$ is a new state not in $Q$. 

The fourth modification concerns configurations 
$qa^n \vdash pa^{kn}$, if $(q,k,p,r)$ is a transition rule of $M$
and $kn$ is an integer. In this case, we store the information of the multiple $k$ as superscript $(k)$ 
in the state $p$ of the successor configuration. In addition,
we store in state~$q$ as subscript the result of $n\bmod (1/k)$, if $k<1$.
That is, $qa^nq pa^{kn}p$ is modified to 
$qa^nq_{(n \bmod (1/k))} p^{(k)}a^{kn}p^{(k)}$ and $q_0'a^{2^{i-1}}q_0' q_0a^{2^i}q_0$ is modified to 
$q_0'a^{2^{i-1}}q_0'q_0^{(2)}a^{2^i}q_0^{(2)}$.  

The fifth modification concerns configurations $qa^n \vdash ra^{n}$, if $(q,k,p,r)$ is a transition rule of~$M$
and $kn$ is not an integer. Hence, $k \in \{1/2,1/3,1/5,1/7\}$ and we store the 
information $(1)$ as superscript in the state $r$ of the successor configuration. In addition,
we store in state~$q$ as subscript the result of $n \bmod (1/k)$.
That is, $qa^nq ra^{n}r$ is modified to 
$qa^nq_{(n \bmod (1/k))} r^{(1)}a^{n}r^{(1)}$.

Finally, if the state $q_{f'}$ has been introduced in the second modification,
we modify $q_{f'}a^pq_{f'} q_{f}a^{p}q_{f}$ to 
$q_{f'}a^pq_{f'} q_{f}^{(1)}a^{p}q_{f}^{(1)}$.

Formally, the set $\valc'(M)$ of valid computations of a multiplying counter machine~$M$
is the set of strings. Let $i=x_1 + 3x_2 + 3^2x_3 + \cdots + 3^{n-1}x_n$ be the unary encoding
of a string $x_1x_2\cdots x_n \in L$. 
If $i >0$, we have
$$
q_0'a'q_0'q_0'a^2q_0' \cdots q_0'a^{2^{i-1}}q_0'q_0^{(\ell_0)}a^{n_0}q^{(\ell_0)}_{0,\varphi_0}
q_1^{(\ell_1)}a^{n_1}q^{(\ell_1)}_{1,\varphi_1} \cdots
q_j^{(\ell_j)}a^{n_j}q^{(\ell_j)}_{j,\varphi_j} 
q_{j+1}^{(\ell_{j+1})}a^{n_{j+1}}q^{(\ell_{j+1})}_{j+1,\varphi_{j+1}} \cdots 
q_f^{(\ell_n)}a^{n_n}q_f^{(\ell_n)}
$$
If $i=0$, we have
$$
q_0^{(\ell_0)}a'q^{(\ell_0)}_{0,\varphi_0}
q_1^{(\ell_1)}a^{n_1}q^{(\ell_1)}_{1,\varphi_1} \cdots
q_j^{(\ell_j)}a^{n_j}q^{(\ell_j)}_{j,\varphi_j} 
q_{j+1}^{(\ell_{j+1})}a^{n_{j+1}}q^{(\ell_{j+1})}_{j+1,\varphi_{j+1}} \cdots 
q_f^{(\ell_n)}a^{n_n}q_f^{(\ell_n)}
$$
where $q_j a^{n_j}$ ($1 \le j \le n$) are configurations of~$M$
such that $n_0=2^i$ and $q_j a^{n_j} \vdash q_{j+1} a^{n_{j+1}}$ for $0 \le j < n$. 
Furthermore, due to the definition of~$M$ we know that for $0 < j \le n$ 
each state $q_j$ different from $q_f$ and $q_{f'}$ has an associated number $k_j \in C$. If $k_j>1$ or, if $k_j<1$ and $n_j$ is divisible by $1/k_j$,
we define $\ell_j=k_{j-1}$ and $\ell_j=1$ otherwise. Moreover, $\ell_0=2$.
For $0 \le j < n$ we define $\varphi_j=n_j \bmod (1/k_j)$, if $k_j<1$, and $\varphi_j=0$ otherwise.
We illustrate the definition with the following example.

\begin{example}
We consider some valid computation $q_0aq_0a^2q_0a^4q_0a^8q_0a^{16}q_1a^{32}q_2a^{16}q_3a^{16}q_4a^{8}q_fa^{8}$ 
of a multiplying counter machine $M$ with transitions
$(q_0,2,q_1,\cdot)$, $(q_1,\frac{1}{2},q_2,\cdot)$, $(q_2,\frac{1}{3},\cdot,q_3)$,
$(q_3,\frac{1}{2},q_4,\cdot)$, and $(q_4,\frac{1}{5},\cdot,q_f)$. States not relevant for the example are denoted by $\cdot$. According to the above discussion we obtain
$$q_0'a'q_0'q_0'a^2q_0'q_0'a^4q_0'q_0'a^8q_0'q_0^{(2)}a^{16}q_{0,0}^{(2)}q_1^{(2)}a^{32}q_{1,0}^{(2)}q_2^{(\frac{1}{2})}a^{16}q_{2,1}^{(\frac{1}{2})}q_3^{(1)}a^{16}q_{3,0}^{(1)}q_4^{(\frac{1}{2})}a^{8}q_{4,3}^{(\frac{1}{2})}q_f^{(1)}a^{8}q_f^{(1)}.$$
\eoe
\end{example}

Our next goal is to represent the set $\valc'(M)$ of such modified valid computations
as the intersection of two languages $\valc'_1(M)$ and $\valc'_2(M)$ that are accepted
by real-time $\revcak{1}$s. To this end, we define $\valc'_1(M)$  
to be the set of strings that start with $q_0'a'q_0'$
or $q_0^{(\ell_0)}a'q^{(\ell_0)}_{0,\varphi_0}$, end with $q_f^{(k)} a^* q_f^{(k)}$
($k \in C \cup \{1\}$),
have no $a'$ or $q_f^{(k)}$ in between, 
and we require that the successor configuration of any configuration
at an odd position is correctly computed.
Similarly, $\valc'_2(M)$ is the set of strings that have the same format as $\valc'_1(M)$
and the successor configuration of any configuration at an even position is correctly
computed.
Due to the required format of $\valc'_1(M)$ and $\valc'_2(M)$
we obtain that $\valc'_1(M) \cap \valc'_2(M)=\valc'(M)$. Moreover, we have
that $\valc'(M)$ is empty if and only if~$M$ accepts the empty set.

\begin{lemma}\label{lem:valc-acc-counter1}
Let $M$ be a multiplying counter machine. Then real-time 
$\revcak{1}$s accepting the sets $\valc'_1(M)$ 
and $\valc'_2(M)$ can effectively be constructed from $M$.
\end{lemma}

\begin{proof}
We describe the construction of a real-time $\revcak{1}$ accepting the set~$\valc'_1(M)$.
A real-time $\revcak{1}$ accepting the set~$\valc'_2(M)$ can similarly be constructed.
First we note that the required correct formatting of the input, namely,
starting with $q_0'a'q_0'$
or $q_0^{(\ell_0)}a'q^{(\ell_0)}_{0,\varphi_0}$, end with $q_f^{(k)} a^* q_f^{(k)}$
($k \in C \cup \{1\}$),
have no $a'$ or $q_f^{(k)}$ in between,  
can be tested by a reversible deterministic finite automaton.
Hence, this test can be realized in an additional component using the standard 
cross product construction. 
(See, e.g.,~\cite{kutrib:2012:rpa}, where the construction for reversible pushdown automata is described.) 

We note that any string in~$\valc'_1(M)$ consists of a sequence of blocks of adjacent configurations 
having one of the following forms:
\begin{enumerate}
\item $q_0'a^nq_0'q_0'a^{2n}q_0'$ for some $n \ge 1$ (if $n=1$, $a^n=a'$),
\item $q_0'a^{n}q_0'q_0^{(2)}a^{2n}q^{(2)}_{0,\varphi_0}$ for some $n \ge 1$ (if $n=1$, $a^n=a'$),
\item $q_j^{(\ell_j)}a^{n_j}q^{(\ell_j)}_{j,\varphi_j} 
q_{j+1}^{(\ell_{j+1})}a^{n_{j+1}}q^{(\ell_{j+1})}_{j+1,\varphi_{j+1}}$, for some $j \ge 0$, 
\item $q_{n-1}^{(\ell_{n-1})}a^{n_{n-1}}q^{(\ell_{n-1})}_{{n-1},\varphi_{n-1}} q_f^{(\ell_n)}a^{n_n}q_f^{(\ell_n)}$ or $q_{n-1}^{(\ell_{n-1})}a^{n_{n-1}}q^{(\ell_{n-1})}_{{n-1},\varphi_{n-1}} q_{f'}^{(\ell_n)}a^{n_n}q_{f'}^{(\ell_n)}$, 
\item $q_{f'}^{(\ell_n)}a^{n_n}q_{f'}^{(\ell_n)}q_{f}^{(1)}a^{n_n}q_{f}^{(1)}$,
\item $q_0^{(\ell_0)}a'q^{(\ell_0)}_{0,\varphi_0} 
q_{1}^{(\ell_{1})}a^{n_{1}}q^{(\ell_{1})}_{1,\varphi_{1}}$, 
\item $q_{0}^{(\ell_{0})}a'q^{(\ell_{0})}_{{0},\varphi_{0}} q_f^{(\ell_1)}a^{n_1}q_f^{(\ell_1)}$ or $q_{0}^{(\ell_{0})}a'q^{(\ell_{0})}_{{0},\varphi_{0}} q_{f'}^{(\ell_1)}a^{n_1}q_{f'}^{(\ell_1)}$. 
\end{enumerate}
We now describe how each such block can be accepted by a quasi real-time $\revcak{1}$.
A quasi real-time $\revcak{1}$ accepting blocks of the first form basically increases the counter for every $a$ (or $a'$ if $n=1$)
from the first part and decreases the counter for every second $a$ from the second part. This can be done reversibly,
since in the backward computation every other~$a$ from the second part increases the counter while every~$a$ (or $a'$ if $n=1$)
from the first part decreases the counter. We note that the counter is empty at the end of each forward computation.
A quasi real-time $\revcak{1}$ accepting blocks of the second form can similarly be constructed.
The basic difference is that after reading $q_0^{(2)}$ in another component of the state set a counter
modulo $1/k_0$ is started, if $k_0<1$. We recall that $k_0 \in C$ is the number associated to state~$q_0$.
While reading $a$'s this counter is updated and finally compared with
$\varphi_0$ when reading $q^{(2)}_{0,\varphi_0}$. If $k_0 \ge 1$, then no counter is started and it is only checked
whether $\varphi_0$ is $0$ when reading $q^{(2)}_{0,\varphi_0}$. This additional behavior can be realized reversibly.

For blocks of the third form we first note that we can reversibly check whether $\varphi_j$ and $\varphi_{j+1}$
are correctly computed with respect to $a^{n_j}$ and $a^{n_{j+1}}$ by adapting the method described for blocks of the 
second form. Let $(q_j,k_j,p,r)$ be the transition rule for state~$q_j$
Now, a quasi real-time $\revcak{1}$ increases its counter for every $a$ from the first part. 
If $\varphi_j=0$, we know that $q_{j+1}=p$ and $\ell_{j+1}=k_j > 1$ or $n_j$ is divisible by $1/k_j$, if $k_j<1$. 
If $k_j > 1$, then we decrease the counter for every $k_j$-th $a$ from the second part.
If $k_j < 1$, then we decrease the counter by $1/k_j$ within $1/k_j$ time steps on every~$a$. This is realized
by $1/k_j-1$ stationary moves on every~$a$.
If $\varphi_j >0$, we know that $q_{j+1}=r$ and $n_j=n_{j+1}$. Hence, we decrease the counter for every $a$ from the second part.
This can be done reversibly, since in the backward computation we know due to the information $\ell_{j+1}$ what to do
on the counter. If $\ell_{j+1}=1$, then the counter is increased for every~$a$ from the second part. 
If $\ell_{j+1}>1$, then the counter is increased for every $k_j$-th $a$ from the second part. 
If $\ell_{j+1}<1$, then the counter is increased by $1/k_j$ within $1/k_j$ time steps for every~$a$ from the second part. 
Subsequently, the counter is decreased for every~$a$ from the first part. Altogether, a quasi real-time $\revcak{1}$
accepting blocks of the third form can be constructed. 
We note that the counter is empty at the end of each forward computation.
A quasi real-time $\revcak{1}$ accepting blocks of the fourth form can similarly be constructed.
The only difference is to replace $q^{(\ell_{j+1})}_{j+1}$ and $q^{(\ell_{j+1})}_{j+1,\varphi_{j+1}}$
by $q^{(\ell_{n})}_{f}$ or $q^{(\ell_{n})}_{f'}$, respectively.
Finally, by using similar ideas we can also construct quasi real-time 
$\revcak{1}$s accepting blocks of the remaining forms.

Since every block can be accepted by a quasi real-time $\revcak{1}$ and the counter is empty at the end of each computation,
we can iterate these automata and obtain a quasi real-time $\revcak{1}$ for~$\valc'_1(M)$
which can be sped-up to a real-time $\revcak{1}$ owing to Theorem~\ref{theo:dcak-quasirealtime}.
\end{proof}

\newpage

\begin{lemma}\label{lem:valc-acc-counter2}
Let $M$ be a multiplying counter machine. Then a real-time 
$\revcak{2}$ accepting the set $\valc'(M)$ can effectively be constructed from $M$.
\end{lemma}

\begin{proof}
We consider the real-time $\revcak{1}$s $M_1$ and $M_2$ constructed in the proof of
Lemma~\ref{lem:valc-acc-counter1} that accept $\valc'_1(M)$ and $\valc'_2(M)$, respectively.
We note that in all accepting computations in $M_1$ as well as in $M_2$ the input has been completely read.
Hence, we can apply the well known Cartesian product technique for intersection
and construct a real-time $\revcak{2}$~$M'$
that simulates $M_1$ in one component of the state set and uses one counter and
and simulates $M_2$ in a second component of the state set and uses the other counter.
The accepting states of $M'$ are defined as $F_1 \times F_2$, where $F_1$ and $F_2$ are the accepting states
in $M_1$ and $M_2$, respectively.
Hence, $M'$ accepts $\valc'_1(M) \cap \valc'_2(M)=\valc'(M)$.
\end{proof}

Now, we have all preparatory results to show the following non-semidecidability
results.

\begin{theorem}\label{theo:undec:counter2}
Let $M$ and $M'$ be two real-time $\revcak{k}$s with $k \ge 2$. Then the following questions are not semidecidable.
\begin{enumerate}
\item Is \mbox{$L(M) = \emptyset$}?
\item Is \mbox{$L(M)$} finite/infinite?
\item Is \mbox{$L(M) \subseteq L(M')$}?
\item Is \mbox{$L(M) = L(M')$}?
\item Is \mbox{$L(M)$} regular/context free?
\end{enumerate}
\end{theorem}

\begin{proof}
Let $T$ be some Turing machine accepting a recursively enumerable set over $\{1,2\}^*$
and $M_0$ its corresponding multiplying counter machine
according to Theorem~\ref{theo:Hartmanis}.
Then, the language $L(T)$ is empty if and only if $M_0$ accepts the empty
set. Moreover, $M_0$ accepts the empty set if and only if $\valc'(M_0)=\emptyset$
and $T$ accepts a finite set if and only if $M_0$ accepts a finite set if and only if 
$\valc'(M_0)$ is finite.

Now, let $M_1$ be a real-time $\revcak{2}$ accepting $\valc'(M_0)$ according to the
construction in Lemma~\ref{lem:valc-acc-counter2}. Hence, $L(M_1)$ is empty,
finite, or infinite if and only if the Turing machine $T$ accepts an empty,
finite, or infinite set. Since the latter questions are not semidecidable
for Turing machines (see, e.g.,~\cite{Hopcroft:1979:itatlc:book}), 
they are not semidecidable for real-time $\revcak{2}$s as well.

It is easy to construct a real-time $\revcak{2}$ accepting the empty set.
If the questions of inclusion and equivalence would be semidecidable, the question of
emptiness would be semidecidable as well which is a contradiction. Hence, both questions are not semidecidable. 

It is described in~\cite{Hopcroft:1979:itatlc:book} how to define the set of
valid computations $\valc(T)$ of a Turing machine~$T$. 
In addition, it is shown there with the help of the pumping lemma that 
$\valc(T)$ is not a context-free language if $T$ accepts an infinite language. 
It can be shown with a similar approach that $\valc'(M_0)$ is not a context-free language
if $M_0$ accepts an infinite language. On the other hand, if $M_0$ accepts a 
finite language, then $\valc'(M_0)$ is finite and hence in particular a regular and
a context-free language. Altogether, we have that $\valc'(M_0)$ is finite if and only
if $M_0$ is regular or context free. If the regularity or context-freeness of
a real-time $\revcak{2}$ would be semidecidable, we would therefore obtain that
the finiteness problem for a real-time $\revcak{2}$ is semidecidable as well
which is a contradiction and shows the remaining claim of the theorem.
\end{proof}

\newpage

\end{document}